\documentclass[10pt]{llncs}
\usepackage{amsmath,amssymb,graphicx,epsfig,color}

\usepackage{fullpage}
\usepackage[english]{babel}
\usepackage{enumerate}

\newcommand{\vct}[1]{}
\newcommand{\mtx}[1]{}

\newcommand{\Id}{\mathbf{I}}
\newcommand{\J}{\mathbf{J}}

\newcommand{\zeromtx}{\mathbf{0}}

\newcommand{\eps}{\varepsilon}

\newcommand{\ignore}[1]{}

\newcommand{\Sym}{\mathcal{S}}
\newcommand{\OO}{\mathcal{O}}

\newcommand{\reals}{\mathbb{R}}
\newcommand{\NN}{\mathbb{N}}

\newcommand{\Prob}[1]{\ensuremath{\mathbb{P}\left(#1\right)}}
\newcommand{\e}{\ensuremath{{\rm e}}}
\DeclareMathOperator{\EE}{\mathbb{E}}

\newcommand{\norm}[1]{\ensuremath{\left\|#1\right\|}}

\newcommand{\infnorm}[1]{\ensuremath{\left\|#1\right\|_\infty}}
\newcommand{\frobnorm}[1]{\ensuremath{\left\|#1\right\|_{\text{\rm F}}}}
\newcommand{\sr}[1]{\ensuremath{\mathrm{\textbf{\small sr}}\left(#1\right)}}
\newcommand{\nnz}[1]{\ensuremath{\mathrm{\textbf{\footnotesize nnz}}\left(#1\right)}}
\newcommand{\trace}[1]{\ensuremath{\mathrm{\textbf{tr}}\left(#1\right)}}
\newcommand{\dil}[1]{\ensuremath{\mathcal{D}\left(#1\right)}}

\newcommand{\expm}[1]{\ensuremath{\mathrm{\textbf{\footnotesize exp}}\left[#1\right]}}
\newcommand{\coshm}[1]{\ensuremath{\mathrm{\textbf{\footnotesize cosh}}\left[#1\right]}}

\newcommand{\sign}[1]{\ensuremath{\mathrm{\textbf{\footnotesize sgn}}\left(#1\right)}}

\newcommand{\argmin}{\operatorname*{arg\; min}}

\newcommand{\diag}[1]{\ensuremath{\mathrm{\textbf{\footnotesize diag}}\left(#1\right)}}

\newcommand{\polylog}[1]{\ensuremath{\mathrm{polylog}\left(#1\right)}}

\newcommand{\Cay}[2]{\ensuremath{\mathrm{Cay}\left(#1; #2\right)}}

\ignore{

\newrefformat{fig}{Fig.~\ref{#1}}
\newrefformat{par}{Section~\ref{#1}}
\newrefformat{sec}{Section~\ref{#1}}
\newrefformat{sub}{Section~\ref{#1}}
\newrefformat{table}{Table~\ref{#1}}
\newrefformat{alg}{Algorithm~\ref{#1}}
\newrefformat{Def}{Definition~\ref{#1}}
\newrefformat{Thm}{Theorem~\ref{#1}}
\newrefformat{step}{Step~\ref{#1}}
\newrefformat{eq}{Equation~\eqref{#1}}
\def\Eqref Eq:#1:{\eqref{eq:#1}}
\newrefformat{Eq}{Equation~\Eqref#1:}
}

\usepackage{float}

\usepackage{hyperref}
\hypersetup{
    unicode=false,          
    pdftoolbar=true,        
    pdfmenubar=true,        
    pdffitwindow=true,      
    pdftitle={A Matrix Hyperbolic Cosine Algorithm and Applications},    
    pdfauthor={Anastasios Zouzias},     
    pdfsubject={},   
    pdfcreator={Anastasios Zouzias},   
    pdfproducer={}, 
    pdfkeywords={derandomization,probability,conditional expectation method, pessimistic estimators}, 
    pdfnewwindow=true,      
    colorlinks=false,       
    linkcolor=red,          
    citecolor=green,        
    filecolor=magenta,      
    urlcolor=cyan          
}

\usepackage{algorithm}
\usepackage{algpseudocode}

\newtheorem{fact}[theorem]{Fact}

\title{A Matrix Hyperbolic Cosine Algorithm and Applications}
\author{Anastasios Zouzias
}
\institute{Department of Computer Science\\ University of Toronto, Canada}

\begin{document}
\maketitle
\begin{abstract}
In this paper, we generalize Spencer's hyperbolic cosine algorithm to the matrix-valued setting. We apply the proposed algorithm to several problems by analyzing its computational efficiency under two special cases of matrices; one in which the matrices have a group structure and an other in which they have rank-one. As an application of the former case, we present a deterministic algorithm that, given the multiplication table of a finite group of size $n$, it constructs an expanding Cayley graph of logarithmic degree in near-optimal $\OO(n^2\log^3 n)$ time. For the latter case, we present a fast deterministic algorithm for spectral sparsification of positive semi-definite matrices, which implies an improved deterministic algorithm for spectral graph sparsification of dense graphs. In addition, we give an elementary connection between spectral sparsification of positive semi-definite matrices and element-wise matrix sparsification. As a consequence, we obtain improved element-wise sparsification algorithms for diagonally dominant-like matrices.

\end{abstract}

\section{Introduction}
A non-trivial generalization of Chernoff bound type inequalities for matrix-valued random variables was introduced by Ahlswede and Winter~\cite{chernoff:matrix_valued:AW}. In parallel, Vershynin and Rudelson introduced similar matrix-valued concentration inequalities using different machinery~\cite{rudelson:isotropic,lowrank:rankone:VR}. Following these two seminal papers, many variants have been proposed in the literature~\cite{recht:simple_completion}; see~\cite{chernoff:matrix_valued:Tropp} for more. Such inequalities, similarly to their real-valued ancestors, provide powerful tools to analyze probabilistic constructions and the performance of randomized algorithms. There is a rapidly growing line of research exploiting the power of these inequalities including new proofs of probabilistic constructions of expander graphs~\cite{expander:AlonRoichman:orig,expander:AlonRoichman:RusLan,chernoff:matrix_valued:Azuma_Naor}, matrix approximation by element-wise sparsification~\cite{matrix:sparsification:IPL2011}, graph approximation via edge sparsification~\cite{graph:sparsifiers:eff_resistance}, analysis of algorithms for matrix completion and decomposition of low rank matrices~\cite{recht:simple_completion,chernoff:matrix_valued:MZ11}, semi-definite relaxation and rounding of quadratic maximization problems~\cite{chernoff:matrix_valued:opt:journal}.
In many settings, it is desirable to convert the above probabilistic proofs into \emph{efficient} deterministic procedures. That is, to derandomize the proofs. Wigderson and Xiao presented an efficient derandomization of the matrix Chernoff bound by generalizing Raghavan's method of pessimistic estimators to the matrix-valued setting~\cite{chernoff:matrix_valued:derand:WX08}. In this paper, we generalize Spencer's hyperbolic cosine algorithm to the matrix-valued setting~\cite{hyperbolic_cosine:Spencer}. In an earlier, preliminary version of our paper~\cite{matrix:hypercosine_zouzias} the generalization of Spencer's hyperbolic cosine algorithm was also based on the method of pessimistic estimators. However, here we present a proof which is based on a simple averaging argument. Next, we carefully analyze two special cases of matrices; one in which the matrices have a group structure and the other in which they have rank-one. We apply our main result to the following problems: deterministically constructing Alon-Roichman expanding Cayley graphs, approximating graphs via edge sparsification and approximating matrices via element-wise sparsification.
The Alon-Roichman theorem asserts that Cayley graphs obtained by choosing a logarithmic number of group elements independently and uniformly at random are expanders~\cite{expander:AlonRoichman:orig}. The original proof of Alon and Roichman is based on Wigner's trace method, whereas recent proofs rely on matrix-valued deviation bounds~\cite{expander:AlonRoichman:RusLan}. Wigderson and Xiao's derandomization of the matrix Chernoff bound implies a deterministic $\OO(n^4 \log n )$ time algorithm for constructing Alon-Roichman graphs. Independently, Arora and Kale generalized the multiplicative weights update (MWU) method to the matrix-valued setting and, among other interesting implications, they improved the running time to $\OO(n^3\polylog{n})$~\cite{phdthesis:Kale:2008}. Here we further improve the running time to $\OO(n^2 \log^3 n)$ by exploiting the group structure of the problem. In addition, our algorithm is combinatorial in the sense that it only requires counting the number of all closed (even) paths of size at most $\OO(\log n)$ in Cayley graphs. All previous algorithms involve numerical matrix computations such as eigenvalue decompositions and matrix exponentiation.
The second problem that we study is the graph sparsification problem. This problem poses the question whether any dense graph can be approximated by a sparse graph under different notions of approximation. Given any undirected graph, the most well-studied notions of approximation by a sparse graph include approximating, \emph{all} pairwise distances up to an additive error~\cite{graph:spanners:PelegS89}, every cut to an arbitrarily small multiplicative error~\cite{graph:sparsifier:BenczurK96} and every eigenvalue of the difference of their Laplacian matrices to an arbitrarily small relative error~\cite{graph:sparsifier:ICM2010}; the resulting graphs are usually called \emph{graph spanners}, \emph{cut sparsifiers} and \emph{spectral sparsifiers}, respectively. Given that the notion of spectral sparsification is stronger than cut sparsification, so we focus on spectral sparsifiers. An efficient randomized algorithm to construct an $(1+\eps)$-spectral sparsifier with $\OO(n\log n /\eps^2)$ edges was given in~\cite{graph:sparsifiers:eff_resistance}. Furthermore, an $(1+\eps)$-spectral sparsifier with $\OO(n/\eps^2)$ edges can be computed in $\OO(mn^3/\eps^2)$ deterministic time~\cite{graph:sparsifiers:twice_ram}. The latter result is a direct corollary of the spectral sparsification of positive semi-definite (psd) matrices problem as defined in~\cite{phdthesis:Srivastava:2010}; see also~\cite{graph:sparsification:Naor} for more applications. Here we present a fast deterministic algorithm for spectral sparsification of psd matrices and, as a consequence, we obtain an improved deterministic spectral graph sparsification algorithm for the case of dense graphs.
The last problem that we analyze is the element-wise matrix sparsification problem. This problem was first introduced by Achlioptas and McSherry in~\cite{matrix:sparsification:optas}. They described sampling-based algorithms that select a small number of entries from an input matrix $A$, forming a sparse matrix $\widetilde{A}$, which is close to $A$ in the operator norm sense. The motivation to study this problem lies on the need to speed up several matrix computations including approximate eigenvector computations~\cite{matrix:sparsification:optas} and semi-definite programming solvers~\cite{fast_SDP:AHK05,Asp09}. Recently, there are many follow-up results on this problem~\cite{matrix:sparsification:arora,matrix:sparsification:IPL2011}. To the best of our knowledge, all known algorithms for this problem are randomized (see Table~$1$ of~\cite{matrix:sparsification:IPL2011}). In this paper we present the first deterministic algorithm and strong sparsification bounds for self-adjoint matrices that have an approximate diagonally dominant\footnote{A self-adjoint matrix $A$ of size $n$ is called \emph{diagonally dominant} if $|A_{ii}| \geq \sum_{j\neq i} |A_{ij}|$ for every $i\in{[n]}$.} property. Diagonally dominant matrices arise in many applications such as the solution of certain elliptic differential equations via the finite element method~\cite{SDD:Vavasis}, several optimization problems in computer vision~\cite{SDD:vision:Koutis} and computer graphics~\cite{SDD:graphics:Joshi}, to name a few.
%
%
\paragraph{Organization of the Paper.}
%
%
The paper is organized as follows. In~\S~\ref{sec:derand_Bernstein}, we present the matrix hyperbolic cosine algorithm (Algorithm~\ref{alg:derand:Bernstein}). We apply the matrix hyperbolic cosine algorithm to derive improved deterministic algorithms for the construction of Alon-Roichman expanding Cayley graphs in~\S~\ref{sec:AR_graphs}, spectral sparsification of psd matrices in~\S~\ref{sec:fast_isotrop_sparse} and element-wise matrix sparsification in \S~\ref{sec:sparsification:matrix}. Due to space constraints, almost all proofs have been deferred to the Appendix.
%
\subsection*{Our Results}
%
The main contribution of this paper is a generalization of Spencer's hyperbolic cosine algorithm to the matrix-valued setting~\cite{hyperbolic_cosine:Spencer},~\cite[Lecture~$4$]{book:probmeth:Spencer}, see Algorithm~\ref{alg:derand:Bernstein}. As mentioned in the introduction, our main result has connections with a recent derandomization of matrix concentration inequalities~\cite{chernoff:matrix_valued:derand:WX08}. We should highlight a few advantages of our result compared to~\cite{chernoff:matrix_valued:derand:WX08}. First, our construction does not rely on composing two separate estimators (or potential functions) to achieve operator norm bounds and does not require knowledge of the sampling probabilities of the matrix samples as in~\cite{chernoff:matrix_valued:derand:WX08}. In addition, the algorithm of~\cite{chernoff:matrix_valued:derand:WX08} requires computations of matrix expectations with matrix exponentials which are computationally expensive, see~\cite[Footnote~$6$, p. $63$]{chernoff:matrix_valued:derand:WX08}. In this paper, we demostrate that overcoming these limitations leads to faster and in some cases simpler algorithms.
Next, we demonstrate the usefulness of the main result by analyzing its computational efficiency under two special cases of matrices. We begin by presenting the following result
\begin{theorem}[Restatement of Theorem~\ref{thm:AR_graphs}]
There is a deterministic algorithm that, given the multiplication table of a group $G$ of size $n$, constructs an Alon-Roichman expanding Cayley graph of logarithmic degree in $\OO(n^2\log^3 n)$ time. Moreover, the algorithm performs only group algebra operations that correspond to counting closed paths in Cayley graphs.
\end{theorem}
To the best of our knowledge, the above theorem improves the running time of all previously known deterministic constructions of Alon-Roichman Cayley graphs~\cite{arora:fast_SDP,chernoff:matrix_valued:derand:WX08,phdthesis:Kale:2008}. Moreover, notice that the running time of the above algorithm is optimal up-to poly-logarithmic factors since the size of the multiplication table of a finite group of size $n$ is $\OO(n^2)$.
In addition, we study the computational efficiency of the matrix hyperbolic cosine algorithm on the case of matrix samples with rank-one. The motivation for studying this special setting is its connection with problems such as graph approximation via edge sparsification as was shown in~\cite{graph:sparsifiers:twice_ram,phdthesis:Srivastava:2010} and matrix approximation via element-wise sparsification as we will see later in this paper. The main result for this setting can be summarized in the following theorem (see \S~\ref{sec:fast_isotrop_sparse}), which improves the $\OO(mn^3 /\eps^2)$ running time of~\cite{phdthesis:Srivastava:2010} when, say, $m = \Omega(n^2)$ and $\eps$ is a constant.
\begin{theorem}
Suppose $0 < \eps < 1$ and $A = \sum_{i=1}^{m} v_i \otimes v_i $ are given, with column vectors $v_i\in\reals^n $. Then there are non-negative real weights $\{s_i\}_{i\leq m}$, at most $ \lceil n/\eps^2 \rceil$ of which are non-zero, such that
\begin{equation*}
	(1-\eps)^3 A \preceq \widetilde{A} \preceq (1+\eps)^3 A,
\end{equation*}
where $\widetilde{A} = \sum_{i=1}^{m}s_i v_i \otimes  v_i$. Moreover, there is a deterministic algorithm which computes the weights $s_i$ in\footnote{The $\widetilde{\OO}(\cdot)$ notation hides $\log\log n$ and $\log\log (1/\eps)$ factors throughout the paper.} $\widetilde{\OO}(mn^2 \log^3 n  /\eps^2 + n^4 \log n /\eps^4)$ time.
\end{theorem}
First, as we have already mentioned the graph sparsification problem can be reduced to spectral sparsification of positive semi-definite matrix. Hence as a corollary of the above theorem (proof omitted, see~\cite{phdthesis:Srivastava:2010} for details), we obtain a fast deterministic algorithm for sparsifying dense graphs, which improves the currently best known $\OO(n^5 /\eps^2 )$ running time for this problem.
\begin{corollary}
Given a weighted dense graph $H=(V,E)$ on $n$ vertices with positive weights and $0< \eps <1$, there is a deterministic algorithm that returns an $(1+\eps)$-spectral sparsifier with $\OO(n/ \eps^2)$ edges in $\widetilde{\OO}(n^4 \log n /\eps^2$ $ \max\{ \log^2 n, 1/\eps^2 \})$ time.
\end{corollary}
Second, we give an elementary connection between element-wise matrix sparsification and spectral sparsification of psd matrices. A direct application of this connection implies strong sparsification bounds for self-adjoint matrices that are close to being \emph{diagonally dominant}. More precisely, we give two element-wise sparsification algorithms for self-adjoint and diagonally dominant-like matrices; in its randomized and the other in its derandomized version (see Table~$1$ of~\cite{matrix:sparsification:IPL2011} for comparison). 
Here, for the sake of presentation, we state our results for diagonally dominant matrices, although the results hold under a more general setting (see \S~\ref{sec:sparsification:matrix} for details).
\begin{theorem}
Let $A$ be any self-adjoint and diagonally dominant matrix of size $n$ and $ 0 < \eps <1$. Assume for normalization that $\norm{A}=1$.
\begin{enumerate}[(a)]
 \item
There is a randomized linear time algorithm that outputs a matrix $\widetilde{A}\in \reals^{n\times n}$ with at most $\OO( n \log n /\eps^2)$ non-zero entries such that, with probability at least $1-1/n$, $\norm{A - \widetilde{A}} \leq \eps.$
\item
There is a deterministic $\widetilde{\OO}( \eps^{-2} \nnz{A} n^2 \log n  \max\{ \log^2 n,1/\eps^2 \} )$ time algorithm that outputs a matrix $\widetilde{A}\in \reals^{n\times n}$ with at most $\OO(n/\eps^2)$ non-zero entries such that $\norm{A - \widetilde{A}} \leq \eps.$
\end{enumerate}
\end{theorem}
%
%
\paragraph{Preliminaries.}
The next discussion reviews several definitions and facts from linear algebra; for more details, see~\cite{book:matrix:Bhatia}. By $[n]$ to be the set $\{1,2,\ldots , n \}$. We denote by $\Sym^{n\times n}$ the set of symmetric matrices of size $n$. Let $x\in\reals^n$, we denote by $\diag{x}$ the diagonal matrix containing $x_1,x_2,\ldots ,x_n$. For a square matrix $M$, we also write $\diag{M}$ to denote the diagonal matrix that contains the diagonal entries of $M$. Let $A$ be an $m\times n$ matrix. $A^{(j)}$ will denote the $j$-th column of $A$ and $A_{(i)}$ the $i$-th row of $A$. We denote $\norm{A}=\max \{ \norm{Ax}~|~\norm{x} =1 \}$, $\infnorm{A} = \max_{i\in{[m]}} \sum_{j\in{[n]}} |A_{ij}|$ and by $\frobnorm{A}=\sqrt{\sum_{i,j}{A_{ij}^2}}$ the Frobenius norm of $A$. Also $\sr{A}:=\frobnorm{A}^2/\norm{A}^2$ is the \emph{stable rank} of $A$ and by $\nnz{A}$ the number of its non-zero entries. The trace of a square matrix $B$ is denoted as $\trace{B}$. We write $\J_n$ for the all-ones square matrices of size $n$. For two self-adjoint matrices $X,Y$, we say that $Y\succeq X$ if and only if $Y-X$ is a positive semi-definite (psd) matrix. Let $x\in\reals^n$, then $x\otimes x$ is the $n\times n$ matrix such that $ (x \otimes x)_{i,j} =x_i x_j$. Given any matrix $A$, its \emph{dilation} is defined as $\dil{A} = \left[\begin{matrix}
        \zeromtx      & A \\
	A^\top & \zeromtx
       \end{matrix}\right].
$
It is easy to see that $\lambda_{\max}(\dil{A}) = \norm{A}$, see e.g.~\cite[Theorem~$4.2$]{book:perturbation:stewart}.
%
\paragraph{Functions of Matrices.}
Here we review some basic facts about the matrix exponential and the hyperbolic cosine function, for more details see~\cite{book:Higham:Matrix_fcn}. All proofs of this section have been deferred to the appendix. The matrix exponential of a self-adjoint matrix $A$ is defined as $\expm{A} = \Id + \sum_{k=1}^{\infty} \frac{A^k}{k!}$. Let $A=Q\Lambda Q^\top$ be the eigendecomposition of $A$. It is easy to see that $\expm{A} = Q \expm{\Lambda} Q^\top$. For any real square matrices $A$ and $B$ of the same size that commute, i.e., $AB=BA$, we have that $ \expm{A+B} = \expm{A}\expm{B}$. In general, when $A$ and $B$ do not commute, the following estimate is known for self-adjoint matrices.
\begin{lemma}\cite{ineq:trace_exp:Golden,ineq:trace_exp:Thompson}\label{lem:ineq:golden_thompson}
For any self-adjoint matrices $A$ and $B$, $\trace{\expm{ A + B}} \leq \trace{\expm{A}\expm{B}}$.
\end{lemma}
\ignore{
\begin{lemma}
For any self-adjoint matrix $A\in{\reals^{n\times n}}$ such that $\zeromtx \preceq A \preceq \Id$, and any $\alpha,\beta \in\reals$,
\[\expm{\alpha A + \beta(\Id-A)} \preceq A \exp(\alpha) + (\Id-A) \exp(\beta).\]
\end{lemma}
The Trotter product formula~\cite{expm:Trotter} states that
\begin{theorem}
Let $A,B$ be two self-adjoint matrices. Then
\[\expm{A+B} = \lim_{l\to \infty} \left(\expm{A/l} \expm{B/l}\right)^l\]
\end{theorem}
}We will also need the following fact about matrix exponential for rank one matrices.
\begin{lemma}\label{lem:expm:outerprod}
	Let $x$ be a non-zero vector in $\reals^n$. Then $ \expm{x \otimes x} = \Id_n + \frac{\e^{\norm{x}^2} - 1}{\norm{x}^2} x\otimes x$.
Similarly, $\expm{-x \otimes x} = \Id_n - \frac{1 - \e^{-\norm{x}^2}}{\norm{x}^2} x \otimes x$.
\end{lemma}
Let us define the \emph{matrix hyperbolic cosine} function of a self-adjoint matrix $A$ as $ \coshm{A} := (\expm{A} + \expm{-A}) /2$. Next, we state a few properties of the matrix hyperbolic cosine.
\begin{lemma}\label{lem:dil_vs_expm}
Let $A$ be a self-adjoint matrix. Then $\trace{\expm{\dil{A}}} = 2 \trace{ \coshm{A} }$.
\end{lemma}
\begin{lemma}\label{lem:coshm_with_proj}
Let $A$ be a self-adjoint matrix and $P$ be a projector matrix that commutes with $A$, i.e., $PA=AP$. Then $\coshm{PA} = P\coshm{A} + \Id - P$.
\end{lemma}
\begin{lemma}\cite[Lemma~$2.2$]{Tsuda05}\label{lem:trace:incr_psd}
For any positive semi-definite self-adjoint matrix $A$ of size $n$ and any two self-adjoint matrices $B,C$ of size $n$, $B\preceq C$ implies $\trace{AB} \leq \trace{AC}$.
\end{lemma}
\vspace*{-4.0ex}
\section{Balancing Matrices: a matrix hyperbolic cosine algorithm}\label{sec:derand_Bernstein}
We briefly describe Spencer's balancing vectors game and then generalize it to the matrix-valued setting~\cite[Lecture~$4$]{book:probmeth:Spencer}. Let a two-player perfect information game between Alice and Bob. The game consists of $n$ rounds. On the $i$-th round, Alice sends a vector $v_i$ with $\infnorm{v_i}\leq 1$ to Bob, and Bob has to decide on a sign $s_i\in{\{\pm 1\}}$ knowing only his previous choices of signs and $\{v_{k}\}_{k < i}$.
At the end of the game, Bob pays Alice $\infnorm{\sum_{i=1}^{n} s_i v_i}$. We call the latter quantity, the \emph{value} of the game.
It has been shown in~\cite{Spencer:balanc_vct} that, in the above limited online variant, Spencer's six standard deviations bound~\cite{sixDeviation:Spencer} does not hold and the best value that we can hope for is $\Omega(\sqrt{n \ln n})$. Such a bound is easy to obtain by picking the signs $\{s_i\}$ uniformly at random. Indeed, a direct application of Azuma's inequality to each coordinate of the random vector $\sum_{i=1}^{n} s_i v_i$ together with a union bound over all the coordinates gives a bound of $\OO(\sqrt{n\ln n})$.
Now, we generalize the balancing vectors game to the matrix-valued setting. That is, Alice now sends to Bob a sequence $\{M_i\}$ of self-adjoint matrices of size $n$ with\footnote{A curious reader may ask him/her-self why the operator norm is the right choice. It turns out the the operator norm is the correct matrix-norm analog of the $\ell_\infty$ vector-norm, viewed as the \emph{infinity} Schatten norm on the space of matrices.} $\norm{M_i}\leq 1$, and Bob has to pick a sequence of signs $\{s_i\}$ so that, at the end of the game, the quantity $\norm{\sum_{i=1}^{n} s_i M_i} $ is as small as possible. Notice that the balancing vectors game is a restriction of the balancing matrices game in which Alice is allowed to send only diagonal matrices with entries bounded in absolute value by one. Similarly to the balancing vectors game, using matrix-valued concentration inequalities, one can prove that Bob has a randomized strategy that achieves at most $\OO(\sqrt{n\ln n})$ w.p. at least $1/2$. Indeed,
\begin{lemma}\label{lem:balanc_mtx}
Let $M_i \in \Sym^{n\times n}$, $\norm{M_i} \leq 1$, $1 \leq i \leq n$. Pick $s_i^*\in{\{\pm 1\} }$ uniformly at random for every $i\in{[n]}$. Then $\norm{\sum_{i=1}^{n} s_i^* M_i} = \OO(\sqrt{  n \ln n})$ w.p. at least $1/2$.
\end{lemma}
Now, let's assume that Bob wants to achieve the above probabilistic guarantees using a \emph{deterministic} strategy. Is it possible? We answer this question in the affirmative by generalizing Spencer's hyperbolic cosine algorithm (and its proof) to the matrix-valued setting. We call the resulting algorithm \emph{matrix hyperbolic cosine} (Algorithm~\ref{alg:derand:Bernstein}). It is clear that this simple greedy algorithm implies a deterministic strategy for Bob that achieves the probabilistic guarantees of Lemma~\ref{lem:balanc_mtx} (set $f_j\sim s_j M_j$, $t=n$ and $\eps = \OO(\sqrt{ \ln n / n})$ and notice that $\gamma,\rho^2$ are at most one).
Algorithm~\ref{alg:derand:Bernstein} requires an extra assumption on its random matrices compared to Spencer's original algorithm. That is, we assume that our random matrices have uniformly bounded their ``matrix variance'', denoted by $\rho^2$. This requirement is motivated by the fact that in the applications that are studied in this paper such an assumption translates bounds that depend quadratically on the matrix dimensions to bounds that depend linearly on the dimensions.
We will need the following technical lemma for proving the main result of this section, which is a Bernstein type argument generalized to the matrix-valued setting~ \cite{chernoff:matrix_valued:Tropp}.
\begin{lemma}\label{lem:bounding_w}
Let $f:[m] \to \Sym^{n\times n}$ with $\norm{f(i)} \leq \gamma$ for all $i\in{[m]}$. Let $X$ be a random variable over $[m]$ such that $\EE{f(X)}=\zeromtx$ and $\norm{\EE f(X)^2 } \leq \rho^2$. Then, for any $\theta >0$, $\norm{\EE [ \expm{ \dil{ \theta f(X)} }]}\ \leq\ \exp\left( \rho^2( \e^{\theta \gamma } -1 - \theta \gamma )/ \gamma^2\right).$ In particular, for any $0 < \eps < 1$, setting $\theta = \eps /\gamma$ implies that $\EE [ \expm{ \dil{ \eps f(X) / \gamma} }] \preceq \e^{\eps^2 \rho^2 / \gamma^2} \Id_{2n}$.
\end{lemma}
%
Now we are ready to prove the correctness of the matrix hyperbolic cosine algorithm.
\vspace*{-4.0ex}
\begin{algorithm}{}
	\caption{Matrix Hyperbolic Cosine}\label{alg:derand:Bernstein}
\begin{algorithmic}[1]
\Procedure{Matrix-Hyperbolic}{$\{f_j\}$, $\eps$, $t$}\Comment{$f_j:[m] \to \Sym^{n\times n}$ as in Theorem~\ref{thm:hypercosine:main}, $0 < \eps < 1$.}
\State Set $\theta = \eps /\gamma $
\For {$i=1$ to $t$}
	\State Compute $x_i^*\in{[m]}$: $ x_i^* = \argmin_{k\in{[m]}}\trace{\coshm{ \theta \sum_{j=1}^{i-1} f_j(x_j^*) + \theta f_i(k) }} $
\EndFor
\State \textbf{Output:} $t$ indices $x_1^*, x_2^*, \ldots ,x_t^*$ such that $\norm{ \frac1{t} \sum_{j=1}^{t} f_j(x_j^*) } \leq \frac{\gamma \ln( 2n)}{t\eps } + \frac{\eps\rho^2}{\gamma} $
\EndProcedure
\end{algorithmic}
\end{algorithm}
\vspace*{-4.0ex}
\begin{theorem}\label{thm:hypercosine:main}
Let $f_j:[m] \to \Sym^{n\times n}$ with $\norm{f_j(i)} \leq \gamma$ for all $i\in{[m]}$ and $j=1,2,\ldots$. Suppose that there exists independent random variables $X_1,X_2,\ldots $ over $[m]$ such that $\EE{f_j(X_j)}=\zeromtx$ and $\norm{\EE f_j(X_j)^2 } \leq \rho^2$. Algorithm~\ref{alg:derand:Bernstein} with input $\{f_j\},\eps, t$ outputs a set of indices $\{x_j^*\}_{j\in{[t]}}$ over $[m]$ such that $ \norm{ \frac1{t}\sum_{j=1}^{t} f_j(x_j^*)} \leq \frac{ \gamma \ln (2n)}{t\eps} +  \frac{\eps \rho^2}{\gamma}.$
\end{theorem}
%
We conclude with an open question related to Spencer's six standard deviation bound~\cite{sixDeviation:Spencer}. Does Spencer's six standard deviation bound holds under the matrix setting? More formally, given any sequence of $n$ self-adjoint matrices $\{M_i\}$ with $\norm{ M_i}\leq 1$, does there exist a set of signs $\{s_i\}$ so that $\norm{ \sum_{i=1}^{n} s_i M_i} = \OO(\sqrt{n})$?
\vspace*{-3.0ex}
\section{Alon-Roichman Expanding Cayley Graphs}\label{sec:AR_graphs}
We start by describing expander graphs. Given a connected undirected $d$-regular graph $H=(V,E)$ on $n$ vertices, let $A$ be its adjacency matrix, i.e., $A_{ij}=w_{ij}$ where $w_{ij}$ is the number of edges between vertices $i$ and $j$. Moreover, let $\widehat{A}=\frac1{d}A$ be its normalized adjacency matrix. We allow self-loops and multiple edges. Let $\lambda_1(\widehat{A}),\ldots ,\lambda_n(\widehat{A})$ be its eigenvalues in decreasing order. We have that $\lambda_1(\widehat{A})=1$ with corresponding eigenvector $\mathbf{1}/\sqrt{n}$, where $\mathbf{1}$ is the all-one vector. The graph $H$ is called a spectral expander if $\lambda(\widehat{A}):=\max_{2\leq j}\{ |\lambda_j(\widehat{A})|\}\leq \eps$ for some positive constant $\eps<1$.
Denote by $m_k=m_k(H):= \trace{A^k}$. By definition, $m_k$ is equal to the number of self-returning walks of length $k$ of the graph $H$. A graph-spectrum-based invariant, recently proposed by Estrada is defined as $EE(A) := \trace{\expm{A}}$~\cite{estrada}, which also equals to $\sum_{k=0}^{\infty} m_k/k!$. For $\theta>0$, we define the \emph{even $\theta$-Estrada index} by $EE_{\text{even}}(A,\theta) := \sum_{k=0}^{\infty}  m_{2k}(\theta A)/(2k)!$.
Now let $G$ be any finite group of order $n$ with identity element $\mathtt{id}$. Let $S$ be a multi-set of elements of $G$, we denote by $S\sqcup S^{-1}$ the symmetric closure of $S$, namely the number of occurrences of $s$ and $s^{-1}$ in $S\sqcup S^{-1}$ equals the number of occurrences of $s\in S$. Let $R$ be the right regular representation\footnote{In other words, represent each group algebra element with a permutation matrix of size $n$ that preserves the group structure. This is always possible due to Cayley's theorem.}, i.e., $(R(g_1)\phi)(g_2) = \phi(g_1 g_2)$ for every $\phi : G \to \reals$ and $g_1,g_2\in G$. The Cayley graph $\Cay{G}{S}$ on a group $G$ with respect to the mutli-set $S\subset G$ is the graph whose vertex set is $G$, and where $g_1$ and $g_2$ are connected by an edge
if there exists $s\in S$ such that $g_2 = g_1 s$ (allowing multiple edges for multiple elements in $S$).
In this section we prove the correctness of the following greedy algorithm for constructing expanding Cayley graphs.
\begin{theorem}\label{thm:AR_graphs}
Algorithm~\ref{alg:estradaAR}, given the multiplication table of a finite group $G$ of size $n$ and $0<\eps<1$, outputs a (symmetric) multi-set $S\subset G$ of size $\OO(\log n /\eps^2)$ such that $\lambda (\Cay{G}{S}) \leq \eps$ in $\OO(n^2\log^3 n /\eps^3)$ time. Moreover, the algorithm performs only group algebra operations that correspond to counting closed paths in Cayley graphs.
\end{theorem}
\vspace*{-3.5ex}
\begin{algorithm}{}
	\caption{Expander Cayley Graph via even Estrada Index Minimization}\label{alg:estradaAR}
\begin{algorithmic}[1]
\Procedure{GreedyEstradaMin}{$G$, $\eps$}\Comment{Multiplication table of $G$, $0<\eps <1$}
\State Set $S^{(0)}=\emptyset$ and $t=\OO(\log n /\eps^2)$
\For {$i=1,\ldots t$ }
	\State Let $g_{*}\in G$ that (approximately) min. the even $\eps/ 2$-Estrada index of $\Cay{G}{S^{(i-1)}\cup g \cup g^{-1}}$ over all $g\in G $  \Comment{Use Lemma~\ref{lem:fastEstrada:Cayley}}
	\State Set $S^{(i)} = S^{(i-1)} \cup g_{*} \cup g_{*}^{-1}$
\EndFor
\State \textbf{Output:} A multi-set $S:=S^{(t)}$ of size $2t$ such that $\lambda(\Cay{G}{S}) \leq \eps$ 
\EndProcedure
\end{algorithmic}
\end{algorithm}
\vspace*{-4.0ex}
%
%
Let $\widehat{A}$ be the normalized adjacency matrix of $\Cay{G}{S\sqcup S^{-1}}$ for some $S\subset G$. It is not hard to see that $ \widehat{A} = \frac1{2|S|} \sum_{s\in S}{ (R(s) + R(s^{-1}))}$. 
We want to bound $\lambda (A)$. Notice that $\lambda(A)=\norm{(\Id - \J/n)A}$. Since we want to analyze the second-largest eigenvalue (in absolute value), we consider $(\Id - \J/n)A = \frac1{|S|} \sum_{s\in S}{ (R(s) + R(s^{-1})) /2} - \J/n.$
Based on the above calculation, we define our matrix-valued function as
\begin{equation}\label{eq:AR:samplesnew}
f(g) := (R(g) + R(g^{-1})) / 2 - \J /n
\end{equation}
for every $g\in G$. The following lemma connects the potential function that is used in Theorem~\ref{thm:hypercosine:main} and the even Estrada index.
\begin{lemma}\label{lem:cosh_Estrada}
Let $S\subset G $ and $A$ be the adjacency matrix of $\Cay{G}{S\sqcup S^{-1}}$. For any $\theta>0$, $\trace{ \coshm{ \theta \sum_{s\in S} f(s) } } = EE_{even} (A,\theta/2)  + 1 - \cosh(\theta |S|).$
\end{lemma}
The following lemma indicates that it is possible to efficiently compute the (even) Estrada index for Cayley graphs with small generating set.
\begin{lemma}\label{lem:fastEstrada:Cayley}
	Let $S\subset G $, $\theta,\delta >0$, and $A$ be the adjacency matrix of $\Cay{G}{S}$. There is an algorithm that, given $S$, computes an additive $\delta$ approximation to $EE(\theta A)$ or $EE_{\text{even}}(A,\theta)$ in $\OO(n|S| \max\{ \log (n/\delta) , 2\e^2 |S| \theta \})$ time.
\end{lemma}
%
%
%
\begin{proof}(of Theorem~\ref{thm:AR_graphs})
By Lemma~\ref{lem:cosh_Estrada}, minimizing the even $\eps/2$-Estrada index in the $i$-th iteration is equivalent to minimizing $\trace{ \coshm{ \theta \sum_{s\in S^{(i-1)}} f(s) +\theta f(g) } }$ over all $g\in G$ with $\theta = \eps$. Notice that $f(g)\in \Sym^{n\times n}$ for $g\in G$, $\EE_{g\in_R{G}}{f(g)} = \zeromtx_n$ since $\sum_{g\in G}R(g) = \J$. It is easy to see that $\norm{f(g)} \leq 2$ and moreover a calculation implies that $\norm{\EE_{g\in_R{G}}{f(g)^2}} \leq 2$ as well. Theorem~\ref{thm:hypercosine:main} implies that we get a multi-set $S$ of size $t$ such that $\lambda (\Cay{G}{ S\sqcup S^{-1}})=\norm{\frac1{|S|} \sum_{s\in S} f(s) }  \leq \eps$. The moreover part follows from Lemma~\ref{lem:fastEstrada:Cayley} with $\delta = \frac{\e^{\eps^2}}{n^c}$ for a sufficient large constant $c>0$. Indeed, in total we incur (following the proof of Theorem~\ref{thm:hypercosine:main}) at most an additive $\ln( \delta n \e^{\eps^2 t}) / \eps$ error which is bounded by $\eps$.
\end{proof}
%
%
%
%
%
%
%
\section{Fast Isotropic Sparsification and Spectral Sparsification}\label{sec:fast_isotrop_sparse}
Let $A$ be an $m\times n$ matrix with $m\gg n$ whose columns are in isotropic position, i.e., $A^\top A = \Id_n$. For $0 <\eps < 1$, consider the problem of finding a small subset of (rescaled) rows of $A$ forming a matrix $\widetilde{A}$ such that $\norm{\widetilde{A}^\top \widetilde{A} - \Id } \leq \eps$. The matrix Bernstein inequality (see~\cite{chernoff:matrix_valued:Tropp}) tells us that there exists such a set with size $\OO(n\log n /\eps^2)$. Indeed, set $f(i)=A_{(i)} \otimes A_{(i)} / p_i - \Id_n$ where $p_i = \norm{A_{(i)}}^2 / \frobnorm{A}^2$. A calculation shows that $\gamma$ and $\rho^2$ are $\OO(n)$. Moreover, Algorithm~\ref{alg:derand:Bernstein} implies an $\OO(mn^4 \log n /\eps^2)$ time algorithm for finding such a set. The running time of Algorithm~\ref{alg:derand:Bernstein} for rank-one matrix samples can be improved to $\OO(mn^3 \polylog{n} /\eps^2)$ by exploiting their rank-one structure. More precisely, using fast algorithms for computing all the eigenvalues of matrices after rank-one updates~\cite{Gu:update}. Next we show that we can further improve the running time by a more careful analysis.
We show how to improve the running time of Algorithm~\ref{alg:derand:Bernstein} to $\OO(\frac{mn^2}{\eps^2} \polylog{n, \frac1{\eps}})$ utilizing results from numerical linear algebra including the Fast Multipole Method~\cite{FMM:CGR} (FMM) and ideas from~\cite{Gu:update}. The main idea behind the improvement is that the trace is invariant under any change of basis. At each iteration, we perform a change of basis so that the matrix corresponding to the previous choices of the algorithm is diagonal. Now, Step $4$ of Algorithm~\ref{alg:derand:Bernstein} corresponds to computing all the eigenvalues of $m$ different eigensystems with special structure, i.e., diagonal plus a rank-one matrix. Such eigensystem can be solved in $\OO(n \polylog{n})$ time using the FMM as was observed in~\cite{Gu:update}. However, the problem now, is that at each iteration we have to represent all the vectors $A_{(i)}$ in the new basis, which may cost $\OO(mn^2)$. The key observation is that the change of basis matrix at each iteration is a Cauchy matrix (see Appendix). It is known that matrix-vector multiplication with Cauchy matrices  can be performed efficiently and numerically stable using FMM. Therefore, at each iteration, we can perform the change of basis in $\OO(mn\polylog{n})$ and $m$ eigenvalue computations in $\OO(mn\polylog{n})$ time. The next theorem states that the resulting algorithm runs in $\OO(mn^2 \polylog{n})$ time (see Appendix for proof).
\begin{theorem}\label{thm:derand:isotrop:fast}
Let $A$ be an $m\times n$ matrix with $A^\top A = \Id_n$, $m\geq n$ and $ 0 < \eps <1$. Algorithm~\ref{alg:fast:isotrop} returns at most $t=\OO(n  \ln n/\eps^2)$ indices $x_1^*,x_2^*,\ldots x_t^*$ over $[m]$ with corresponding scalars $s_1,s_2,\ldots ,s_t$ using $\widetilde{\OO}(mn^2 \log^3 n /\eps^2 )$ operations such that
\begin{equation}\label{eq:main_thm:fast:main_eqn}	
	\norm{ \sum_{i=1}^{t} s_i A_{(x_i^*)} \otimes A_{(x_i^*)} - \Id_n} \leq \eps.
\end{equation}
\end{theorem}
%
%
%
%
\vspace*{-3.5ex}
\begin{algorithm}{}
	\caption{Fast Isotropic Sparsification}\label{alg:fast:isotrop}
\begin{algorithmic}[1]
\Procedure{Isotrop}{$A$, $\eps$} \Comment{$A\in{\reals^{m\times n}}$, $\sum_{k=1}^{m}A_{(k)}\otimes A_{(k)} = \Id_n$ and $0 < \eps <1$}
\State Set $\theta = \eps / n $, $t=\OO( n \ln n/\eps^2)$, and $A_{(k)} \leftarrow A_{(k)}/\sqrt{p_k}$ for every $k\in{[m]}$, where $p_k=\norm{A_{(k)}}^2/n$
\State Set $\Lambda_{\{0\}} = \zeromtx_n$ and $Z = \sqrt{\theta}\ A$
\For {$i=1$ to $t$}
	\State $x_i^* = \argmin_{k\in{[m]}}{\trace{\expm{ \Lambda_{\{i - 1\}} + Z_{(k)} \otimes Z_{(k)} }\e^{-\theta i}  + \expm{- \Lambda_{\{i - 1\}} - Z_{(k)} \otimes Z_{(k)} }\e^{\theta i}  }    } $ \Comment{Apply $m$ times Lemma~\ref{lem:comp_eigs}}
	\State $[\Lambda_{\{i\}}, U_{\{i\}}] = \textbf{eigs} ( \Lambda_{\{i - 1\}} + Z_{(x_i^*)} \otimes Z_{(x_i^*)} )$ \Comment{\textbf{eigs} computes eigensystem}
	\State $Z = Z  U_{\{i\}} $ \Comment{Apply fast matrix-vector multiplication }
\EndFor
\State \textbf{Output:} $t$ indices $x_1^*, x_2^*, \ldots ,x_t^*,\ x_i^* \in{[m]}$  s.t. $\norm{ \sum_{k=1}^{t} \frac{A_{(x_k^*)} \otimes A_{(x_k^*)} }{ tp_{x_k^*}} - \Id_n } \leq \eps $
\EndProcedure
\end{algorithmic}
\end{algorithm}
\vspace*{-4.0ex}
%
%
%
%
%
Next, we show that Algorithm~\ref{alg:fast:isotrop} can be used as a bootstrapping procedure to improve the time complexity of~\cite[Theorem~3.1]{phdthesis:Srivastava:2010}, see also~\cite[Theorem~$3.1$]{graph:sparsifiers:twice_ram}. Such an improvement implies faster algorithms for constructing graph sparsifiers and, as we will see in \S~\ref{sec:sparsification:matrix}, element-wise sparsification of matrices.
\begin{theorem}\label{thm:sparsification:here}
Suppose $0 < \eps < 1$ and $A = \sum_{i=1}^{m} v_i \otimes v_i $ are given, with column vectors $v_i\in\reals^n $ and $m\geq n$. Then there are non-negative weights $\{s_i\}_{i\leq m}$, at most $ \lceil n/\eps^2 \rceil$ of which are non-zero, such that
\begin{equation}
	(1-\eps)^3 A \preceq \widetilde{A} \preceq (1+\eps)^3 A,
\end{equation}
where $\widetilde{A} = \sum_{i=1}^{m}s_i v_i \otimes v_i$. Moreover, there is an algorithm that computes the weights $\{s_i\}_{i\leq m}$ in deterministic $\widetilde{\OO}(mn^2 \log^3 n  /\eps^2 + n^4 \log n /\eps^4)$ time.
\end{theorem}
%
%
%
\section{Element-wise Matrix Sparsification}\label{sec:sparsification:matrix}
%
%
%
A deterministic algorithm for the element-wise matrix sparsification problem can be obtained by derandomizing a recent result whose analysis is based on the matrix Bernstein inequality~\cite{matrix:sparsification:IPL2011}.
\begin{theorem}\label{thm:matrix_sparse:slow}
Let $A$ be an $n\times n$ matrix and $ 0 < \eps <1$.  There is a deterministic polynomial time algorithm that, given $A$ and $ \eps$, outputs a matrix $\widetilde{A}\in \reals^{n\times n}$ with at most $ 28 n \ln (\sqrt{2n} ) \sr{A} /\eps^2$ non-zero entries such that $\norm{A - \widetilde{A}} \leq \eps \norm{A}.$
\end{theorem}
Next, we give two improved element-wise sparsification algorithms for self-adjoint and diagonally dominant-like matrices; one of them is randomized and the other is its derandomized version. Both algorithms share a crucial difference with all previously known algorithms for this problem; during their execution they may densify the diagonal entries of the input matrices. On the one hand, there are at most $n$ diagonal entries, so this does not affect asymptotically their sparsity guarantees. On the other hand, as we will see later this twist turns out to give strong sparsification bounds.
Recall that the results of~\cite{graph:sparsifiers:eff_resistance,graph:sparsifiers:twice_ram} imply an element-wise sparsification algorithm that works only for Laplacian matrices. It is easy to verify that Laplacian matrices are also diagonally dominant. Here we extend these results to a wider class of matrices (with a weaker notion of approximation). The diagonally dominant assumption is too restrictive and we will show that our sparsification algorithms work for a wider class of matrices. To accommodate this, we say that a matrix $A$ is $\theta$-symmetric diagonally dominant (abbreviate by $\theta$-SDD) if $A$ is self-adjoint and the inequality $\infnorm{A} \leq \sqrt{\theta} \norm{A}$ holds.
By definition, any diagonally dominant matrix is also a $4$-SDD matrix. On the other extreme, every self-adjoint matrix of size $n$ is $n$-SDD since the inequality $\infnorm{A}\leq \sqrt{n} \norm{A}$ is always valid. The following elementary lemma gives a connection between element-wise matrix sparsification and spectral sparsification as defined in~\cite{phdthesis:Srivastava:2010}.
\begin{lemma}\label{lem:sparsif:decomp}
Let $A$ be a self-adjoint matrix of size $n$ and $R=\diag{R_1,R_2,\ldots ,R_n}$ where $R_i = \sum_{j\neq i} |A_{ij}|$. Then there is a matrix $C$ of size $n\times m$ with $m \leq \binom{n}{2}$ such that
\begin{eqnarray}\label{eqn:sparsify_lemma}
 A = CC^\top +\diag{A} - R.
\end{eqnarray}
Moreover, each column of $C$ is indexed by the ordered pairs $(i,j)$, $i<j$ and equals to $C^{(i,j)} = \sqrt{|A_{ij}|} e_i + \sign{A_{ij}}  \sqrt{|A_{ij}|} e_j$ for every $i<j$, $i,j\in[n]$.
\end{lemma}
\begin{remark}
In the special case where $A$ is the Laplacian matrix of some graph, the above decomposition is precisely the vertex-edge decomposition of the Laplacian matrix, since in this case $\diag{A} =R$.
\end{remark}
Using the above lemma, we give a randomized and a deterministic algorithm for sparsifying $\theta$-SDD matrices. First we present the randomized algorithm.
\begin{theorem}\label{thm:matrix_sparsif:rand}
Let $A$ be a $\theta$-SDD matrix of size $n$ and $ 0 < \eps <1$.  There is a randomized linear time algorithm that, given $A$, $\norm{A}$ and $\eps$, outputs a matrix $\widetilde{A}\in \reals^{n\times n}$ with at most $\OO( n\theta \log n /\eps^2)$ non-zero entries such that w.p. at least $1-1/n$, $\norm{A - \widetilde{A}} \leq \eps \norm{A}.$
\end{theorem}
Next we state the derandomized algorithm of the above result.
\begin{theorem}\label{thm:matrix_sparsif:det}
Let $A$ be a $\theta$-SDD matrix of size $n$ and $ 0 < \eps <1/2$.  There is an algorithm that, given $A$ and $ \eps$, outputs a matrix $\widetilde{A}\in \reals^{n\times n}$ with at most $\OO( n \theta /\eps^2)$ non-zero entries such that $\norm{A - \widetilde{A}} \leq \eps \norm{A}$. Moreover, the algorithm computes $\widetilde{A}$ in deterministic $\widetilde{\OO}(\nnz{A}n^2 \theta\log^3 n  /\eps^2 + n^4 \theta^2 \log n /\eps^4)$ time.
\end{theorem}
\begin{remark}
The results of~\cite{graph:sparsifiers:twice_ram,phdthesis:Srivastava:2010} imply a deterministic $\OO(\nnz{A} \theta n^3 /\eps^2 )$ time algorithm that outputs a matrix $\widetilde{A}$ with at most $ \lceil 19(1+\sqrt{\theta})^2 /\eps^2\rceil n $ non-zero entries such that $\norm{\widetilde{A}-A} \leq \eps\norm{A}$.
\end{remark}
%
%
%
%
\subsection*{Acknowledgements}
%
%
%
%
The author would like to thank Mark Braverman for several interesting discussions and comments about this work.
%
%
{

}
%
%
\section*{Appendix}
\section*{Fast Multiplication with Cauchy Matrices and Special Eigensystems}
We start by defining the so-called Cauchy (generalized Hilbert) matrices. An $m\times n$ matrix $C$ defined by
\[ C_{i,j} := \frac1{t_i - s_j},\quad i\in{[m]},j\in{[n]},\]
where $t=(t_1,\ldots ,t_m),\ t\in\reals^m$ and $s=(s_1,\ldots ,s_n),\ s\in\reals^n$  and $t_i\neq s_j$ for all $i\in{[m]}$ and $j\in{[n]}$ is called \emph{Cauchy}. Given a vector $x\in\reals^n$, the naive algorithm for computing the matrix-vector product $C x$ requires $\OO(mn)$ operations. It is not clear if it is possible to perform this computation in less than $\OO(mn)$ operations. Surprisingly enough, it is possible to compute this product with $\OO((m+n)\log^2 (m+n))$ operations. This computation can be done by two different approaches. The first one is based on fast polynomial multiplication, polynomial interpolation and polynomial evaluation at distinct points~\cite[Algorithm~$1$, p.~$130$]{book:fast_matrix:Bini_Pan}. The main drawback of this approach is its numerical instability. The second approach is based on the so-called Fast Multipole Method (FMM) introduced in~\cite{FMM:CGR}. This method returns an approximate solution to the matrix-vector product for any given error parameter\footnote{That is, given an $n\times n$ Cauchy matrix, a vector $x\in\reals^n$ and $0<\eps< 1$, returns a vector $y\in\reals^n$ so that $\infnorm{y - Cx} \leq \eps$ in time $\OO(n \log^2 (1/\eps))$. In an actual implementation, setting $\eps$ to be a small constant relative to the machine's (numerical) precision suffices; see~\cite[\S~$3$]{Gu:update} for a more careful implementation and discussion on numerical issues.}. Ignoring numerical issues that are beyond the scope of this work, we summarize our discussion to the following 
\begin{lemma}\cite{book:fast_matrix:Bini_Pan,FMM:CGR}\label{lem:fast_mm:gerasoulis}
Let $x\in\reals^n$ and $C$ be a Cauchy matrix defined as above with $t\in\reals^m, s\in\reals^n$. There is an algorithm that, given vectors $s,t,x$, computes the product $C  x$ using $\OO((m+n)\log^2 (m+n))$ operations.
\end{lemma}
%
Given a self-adjoint matrix $B = \Sigma + \rho u \otimes u$, where $\Sigma = \diag{\sigma_1 ,\ldots ,\sigma_n}$, $\rho >0$ and $u\in\reals^n$ is a unit vector, our goal is to efficiently compute all the eigenvalues of $B$. It is well-known that the eigenvalues of $B$ are the roots of a special function, known as secular function~\cite{rank_one_update:Golub} and are interlaced with $\{\sigma_{i}\}_{i\leq n}$. In addition, evaluating the secular function requires $\OO(n)$ operations implying that a standard (Newton) root-finding procedure requires $\OO(n)$ operations per each eigenvalue. Hence, $\OO(n^2)$ operations are required for all eigenvalues. In their seminal paper~\cite{Gu:update}, Gu and Eisenstat showed that it is possible to encode the updates of the root-finding procedure for \emph{all} eigenvalues as matrix-vector multiplication with an $n\times n$ Cauchy matrix. Based on this observation, they showed how to use the Fast Multipole Method for approximately computing all the eigenvalues of this special type of eigenvalue problem.
\begin{lemma}\cite{Gu:update}\label{lem:comp_eigs}
Let $b\in\NN$, $\rho>0$, $\Sigma=\diag{\sigma_1,\sigma_2,\ldots , \sigma_n}$ and $u\in\reals^n$ be a unit vector. There is an algorithm that given $\Sigma, \rho, u$ computes all the eigenvalues of $B=\Sigma + \rho u \otimes u$ within an additive error $2^{-b}\norm{B}$ in $\OO(n \log^2 n \log b )$ operations.
\end{lemma}
%
%
%
%
\section*{Omitted Proofs}
%
%
\begin{proof}(of Lemma~\ref{lem:expm:outerprod})
The proof is immediate by the definition of the matrix exponential. Notice that $(x \otimes x)^k = \norm{x}^{2(k-1)} x\otimes x$ for $k\geq 1$.
\begin{eqnarray*}
		\expm{x \otimes x}  & = & \Id + \sum_{k=1}^{\infty}{ \frac{(x \otimes x)^k}{k!}} \  =\  \Id + \sum_{k=1}^{\infty}{ \frac{\norm{x}^{2(k-1)} x\otimes x}{k!}}
				\ignore{\ = \ \Id + \left(\sum_{k=1}^{\infty}{ \frac{\norm{x}^{2k}}{k!}}\right) \frac{xx^\top}{\norm{x}^2} \\
				& = & \Id + \left(\sum_{k=1}^{\infty}{ \frac{\norm{x}^{2k}}{k!}}\right) \frac{xx^\top}{\norm{x}^2} }
				\ = \ \Id + \frac{\e^{\norm{x}^2} - 1}{\norm{x}^2} x\otimes x.
\end{eqnarray*}
Similar considerations give that $\expm{-x\otimes x } = \Id - \frac{1 - \e^{-\norm{x}^2}}{\norm{x}^2} x \otimes x$.
\end{proof}
%

%
\begin{proof}(of Lemma~\ref{lem:dil_vs_expm})
Set $B :=\dil{A} = \left[ \begin{matrix} \zeromtx & A \\
 A^\top & \zeromtx
\end{matrix}\right]$. Notice that for any integer $k\geq 1$, $B^{2k} = \left[ \begin{matrix}
 A^{2k} & \zeromtx \\
 \zeromtx & A^{2k}
\end{matrix}\right]$ and $
B^{2k+1} = \left[ \begin{matrix}
 \zeromtx & A^{2k+1} \\
 A^{2k+1} & \zeromtx
\end{matrix}\right]$. Since the odd powers of $B$ are trace-less, it follows that
\begin{eqnarray*}
      \trace{ \expm{ B }} &  =  &  \trace{ \Id_{2n}  + \sum_{k=1}^{\infty}  \frac{B^{2k}}{(2k ) !}    + \sum_{k=0}^{\infty}  \frac{B^{2k+1}}{(2k + 1 ) !} }
			  \  =  \  \trace{ \Id_{2n}  + \sum_{k=1}^{\infty}  \frac{B^{2k}}{(2k ) !} } \\
			  &  =  &  2 \trace{ \Id_{n}  + \sum_{k=1}^{\infty}  \frac{A^{2k}}{(2k ) !} }
			  \  =  \  \trace{ \expm{A} + \expm{- A}  }
			  \  =  \  2\trace{ \coshm{A}}.
\end{eqnarray*}
\end{proof}
%
%
\begin{proof}(of Lemma~\ref{lem:coshm_with_proj})
By the definition of $\coshm{\cdot}$, it suffices to show that $\expm{PA}=P\expm{A}+ \Id - P$, 
\begin{eqnarray*}
	\expm{PA}= \Id + \sum_{k=1}^{\infty} \frac{(PA)^k}{k!} = \Id + P\sum_{k=1}^{\infty} \frac{A^k}{k!} = P\expm{A}+ \Id - P.
\end{eqnarray*}
\end{proof}
%
\begin{proof}(of Lemma~\ref{lem:balanc_mtx})
We wish to apply matrix Azuma's inequality, see~\cite[Theorem~$7.1$]{chernoff:matrix_valued:Tropp}. For every $j\in{[n]}$, define the matrix-valued difference sequence $f_j: [2] \to \Sym^{n\times n}$ as $f_j(k) = (2(k -1) -1 )M_j $ with $\norm{f_j(\cdot)} \leq 1$. Let $X$ be a uniform random variable over the set $[2]$. Then $\EE_X f_j(X)= \zeromtx_n$. Set $\eps = \sqrt{10\ln (4 n) / n}$. Matrix-valued Azuma's inequality tells us that w.p. at least $1/2$, a random set of signs $\{s_j\}_{j\in{[n]}}$ satisfies $\norm{\frac1{n} \sum_{j=1}^{n} s_j M_j } \leq \eps$. Rescale the last inequality to conclude.
\end{proof}
%
\begin{proof}(of Theorem~\ref{thm:hypercosine:main})
Using the notation of Algorithm~\ref{alg:derand:Bernstein}, for every $i=1,2,\ldots , t$, define recursively $W(i) := \theta \sum_{j=1}^{i} f_j(x_j^*)$ and the potential function $\Phi^{(i)} := 2\trace{\coshm{W(i)}}$. For all steps $i=1,2,\ldots , t$, we will prove that
\begin{eqnarray}\label{ineq:barrier_incr}
  \Phi^{(i)}   & \leq & \Phi^{(i-1)} \exp\left( \eps^2 \rho^2/\gamma^2  \right).
\end{eqnarray}
Assume that the algorithm has fixed the first $(i-1)$ indices $x_1^*,\ldots ,x_{(i-1)}^* $. An averaging argument applied on the expression of the argmin of Step $4$ gives that
\begin{eqnarray*}
\EE_{X_i}  2\trace{\coshm{ \theta  W(i - 1) + \theta  f_i(X_i)}} &   =  & \EE_{X_i}  \trace{\expm{ \theta \dil{ W(i - 1)} + \theta \dil{ f_i(X_i)} }} \\
                                                                   & \leq &   \trace{\expm{ \dil{\theta W(i - 1)}} \EE_{X_i} \expm{\dil{ \theta f_i(X_i)} }} \\
                                                                   & \leq &   \trace{\expm{ \dil{\theta W(i - 1)}}  \Id_{2n}} \exp\left( \eps^2 \rho^2 / \gamma^2 \right) \\
                                                                   &   =  &   \Phi^{(i-1)} \exp\left( \eps^2 \rho^2 / \gamma^2 \right)
\end{eqnarray*}
where in the first inequality we used Lemma~\ref{lem:dil_vs_expm} and linearity of dilation, in the second inequality we used the Golden-Thompson inequality (Lemma~\ref{lem:ineq:golden_thompson}) and linearity of trace, in the third inequality we used Lemma~\ref{lem:trace:incr_psd} together with Lemma~\ref{lem:bounding_w} and in the last equality we used again Lemma~\ref{lem:dil_vs_expm}. Since the algorithm seeks the minimum of the expression in Step $4$, it follows that $\Phi^{(i)} \leq \EE_{X_i}  2\trace{\expm{ \theta \dil{ W(i-1)} + \theta \dil{ f_i(X_i)} }}$ which proves Ineq.~\eqref{ineq:barrier_incr}. Apply $t$ times Ineq.~\eqref{ineq:barrier_incr} to conclude that $\Phi^{(t)} \leq \Phi^{(0)} \exp\left( t\frac{\eps^2 \rho^2}{ \gamma^2} \right).$
Recall that $\Phi^{(0)} = 2\trace{\coshm{\zeromtx_n}} = 2\trace{\Id_n }  = 2n$. On the other hand, we can lower bound $\Phi^{(t)}$
\[\Phi^{(t)} = 2\trace{\coshm{\theta \sum_{j=1}^{t}  f_j(x_j^*)} }  \geq \exp\left(\norm{ \theta \sum_{j=1}^{t}  f_j(x_j^*)  }\right). \]
The last inequality follows since $2\trace{\coshm{ C}} = 2\sum_{i=1}^{n} \cosh ( \lambda_i ( C )) \geq 2\cosh\left(\lambda_{\max}( C )\right) + 2\cosh\left( \lambda_{\min}(C)\right)  \geq \exp (\norm{C})$ for any matrix $C\in\Sym^{n\times n}$ . Take logarithms on both sides and divide by $\theta$, we conclude that $\norm{ \sum_{j=1}^{t}  f_j(x_j^*) } \leq \frac{\ln (2n)}{\theta} + t\frac{\eps^2 \rho^2}{\theta \gamma^2}$.
Rescale by $t$ the last inequality to conclude the proof.
%
%
\end{proof}
%
%
\begin{proof}(of Lemma~\ref{lem:cosh_Estrada})
For notational convenience, set $P:= \Id_n - \J_n/n$ and $B := \frac{\theta}{2} \sum_{s\in S} (R(s) + R(s)^{-1})  $. Since $\J R(g) = R(g) \J = \J$, we have that $\trace{ \coshm{ \theta \sum_{s\in S} f(s) } } = \trace{ \coshm{ P B }}$. Now using Lemma~\ref{lem:coshm_with_proj}, it follows $\trace{ \coshm{ P B }} = \trace{P \coshm{B} + \Id - P} = \trace{\coshm{B}} + \trace{-\frac{\J}{n} \coshm{B} + \Id - P}$
Notice that $\J/n$ is a projector matrix, hence applying Lemmas~\ref{lem:expm:outerprod},\ref{lem:coshm_with_proj} we get that \[\trace{-\frac{\J}{n} \coshm{B} + \Id - P} = \trace{-\coshm{\J/n B} + P +\Id - P} = 1 - \cosh(\theta |S|).\]
\end{proof}
%
\begin{proof}(of Lemma~\ref{lem:fastEstrada:Cayley})
We will prove the Lemma for $EE(A,\theta)$, the other case is similar. Let $h:=\theta \sum_{s\in S} s$ be a group algebra element of $G$, i.e, $h\in \reals [G]$. Define $\expm{h} := \mathtt{id} + \sum_{k=1}^{\infty} \frac{h^{\star k}}{k!}$ and $T_l(h):= \mathtt{id} + \sum_{k=1}^{l} \frac{h^{\star k}}{k!}$ (where $h^{\star k}$ is the $k$-folded convolution/multiplication over $\reals [G]$) the exponential operator and its $l$ truncated Taylor series, respectively. Notice that $\theta A=\theta \sum_{s\in S} R(s) = R(h) $, so $EE(A,\theta) = \trace{\expm{R(h)}}= \trace{R(\expm{h})}$. We will show that the quantity $\trace{R(T_l(h) )}$ is a $\delta$ approximation for $EE(A,\theta)$ when $l\geq  \max\{ \log (n/\delta) , 2\e^2 |S| \theta\}$.

Compute the sum of $T_l(h)$ by summing each term one by one and keeping track of all the coefficients of the group algebra elements. The main observation is that at each step there are at most $n$ such coefficients since we are working over $\reals [G]$. For $k > 1$, compute the $k$-th term of the sum by $(\sum_{s\in S} c_s s)^k /k! = (\sum_{s\in S} c_s s)^{k - 1 }/(k - 1 )! \cdot \sum_{s\in S} (c_s/k) s.$
	Assume that we have computed the first term of the above product, which is some group algebra element denote it by $\sum_{g\in G} \beta_g g$ for some $\beta_g\in\reals$. Hence, at the next iteration, we have to compute the product/convolution of $\sum_{g\in G} \beta_g g$ with $\theta /k \sum_{s\in S} s$, which can be done in  $\OO( n |S|)$ time. Since the sum has $l$ terms, in total we require $\OO(n|S| l)$ operations. Now, we show that it is a $\delta$ approximation. We need the following fact (see~\cite[Theorem~$10.1$,~p.~$234$]{book:Higham:Matrix_fcn})
\begin{fact}\label{fact:expm:taylor_exp}
For any $B\in \reals^{n\times n}$, let $T_{l}(B) := \sum_{k=0}^{l} \frac{B^k}{k!}$. Then, $ \norm{ \expm{B} - T_{l}(B) } \leq \frac{\norm{B}^{l+1}}{(l+1)!} \e^{\norm{B}}.$
\end{fact}
Notice that $ \norm{\theta A} = \norm{\sum_{s\in S}\theta R(s)} \leq \theta |S|$ by triangle inequality and the fact that $\norm{R(g)}=1$ for any $g\in G$. Applying Fact~\ref{fact:expm:taylor_exp} on $\theta A$ we get that
	\begin{eqnarray*}
		\norm{\expm{\theta A} - T_l(\theta A)}	& \leq & \frac{(\theta |S|)^{l+1}}{(l+1)!}\e^{\theta |S|}\ \leq\ \left(\frac{ \e\theta |S|}{l+1}\right)^{l+1} \e^{\theta |S|} \\
									&   =  & \left(\frac{ \e^{1+ (\theta |S|)/(l+1)}\theta |S|}{l+1}\right)^{l+1}  \leq \frac1{2^{l+1}} \leq \frac{\delta}{n}.
	\end{eqnarray*}
	where we used the inequality $(l+1)! \geq  (\frac{l+1}{e})^{l+1}$ and the assumption that $l\geq \max\{ \log (n/\delta) , 2 \e^2 \theta |S|\}$.
\end{proof}
%
%
\begin{lemma}\label{lem:technical_cosh}
Assume that the first $(i-1)$ indices, $i< t$ have been fixed by Algorithm~\ref{alg:fast:isotrop}. Let $\Phi_k^{(i)}$ be the value of the potential function when the index $k$ has been selected at the next iteration of the algorithm. Similarly, let $\widetilde{\Phi}_k^{(i)}$ be the (approximate) value of the potential function computed using Lemma~\ref{lem:comp_eigs} within an additive error $\delta>0$ for all eigenvalues. Then,
\begin{eqnarray*}
	\e^{-\delta} \Phi_k^{(i)} \leq \widetilde{\Phi}^{(i)}_k \leq  \e^{\delta} \Phi_k^{(i)}
\end{eqnarray*}
\end{lemma}
%
%
\begin{proof}
Let $\tau_1,\tau_2 , \ldots , \tau_n$ be the eigenvalues of $\Lambda_{\{i -1\}} + Z_{(k)} \otimes Z_{(k)}$. Let $\widetilde{\tau}_1, \widetilde{\tau}_2 ,\dots ,\widetilde{\tau}_n$ be the approximate eigenvalues of the latter matrix when computed via Lemma~\ref{lem:comp_eigs} within an additive error $\delta>0 $, i.e, $|\widetilde{\tau}_j - \tau_j| \leq \delta$ for all $j\in{[n]}$.

First notice that, by Step $5$ of Algorithm~\ref{alg:fast:isotrop}, $\Phi_k^{(i)} = 2 \sum_{j=1}^{n} \cosh (\tau_j - \lambda i)$. Similarly, $\widetilde{\Phi}_k^{(i)}:= 2 \sum_{j=1}^{n} \cosh (\widetilde{\tau}_j - \lambda i)$. By the definition of the hyperbolic cosine, we get that
\begin{eqnarray*}
	\sum_{j=1}^{n} \cosh (\widetilde{\tau}_j - \lambda i )   & = &  \sum_{j=1}^{n} \cosh (\tau_j - \lambda i  + \widetilde{\tau}_j - \tau_j )  \\
	& = & \frac1{2}\sum_{j=1}^{n} \left[\exp (\tau_j - \lambda i)\exp(\widetilde{\tau}_j - \tau_j ) + \exp (-\tau_j + \lambda i)\exp(-\widetilde{\tau}_j + \tau_j )\right].
\end{eqnarray*}
To derive the upper bound notice that $\sum_{j=1}^{n} \cosh (\widetilde{\tau}_j - \lambda i )  \leq  \sum_{j=1}^{n} \cosh (\tau_j - \lambda i) \max_{j\in{[n]}}\{ \exp(\widetilde{\tau}_j - \tau_j), \exp ( - \widetilde{\tau}_j + \tau_j) \}$
and the maximum is upper bounded by $\exp(\delta)$. Similarly, for the lower bound.
\end{proof}
%
%
%
\begin{proof}(of Theorem~\ref{thm:derand:isotrop:fast})
The proof consists of three steps: (\emph{a}) we show that Algorithm~\ref{alg:fast:isotrop} is a reformulation of Algorithm~\ref{alg:derand:Bernstein}; (\emph{b}) we prove that in Step $5$ of Algorithm~\ref{alg:fast:isotrop} it is enough to compute the values of the potential function within a sufficiently small multiplicative error using Lemma~\ref{lem:comp_eigs}, and (\emph{c}) we give the advertised bound on the running time of Algorithm~\ref{alg:fast:isotrop}.
Set $p_i = \norm{A_{(i)}}^2/\frobnorm{A}^2$, $f(i) = A_{(i)}\otimes A_{(i)}/p_i - \Id_n$ and $s_i=1/p_i$ for every $i\in{[m]}$. Observe that $\frobnorm{A}^2=\trace{A^\top A} = \trace{\Id_n} = n$. Let $X$ be a random variable distributed over $[m]$ with probability $p_i$. Notice that $\EE{f(X)} = \zeromtx_n$ and $\gamma = n $, since $\norm{f(i)} = \norm{n A_{(i)}\otimes A_{(i)}/\norm{A_{(i)}}^2 - \Id_n } \leq n $ for every $i\in{[m]}$. Moreover, a direct calculation shows that $\EE{f(X)^2} = \EE{ (A_{(X)} \otimes A_{(X)}/p_X)^2} - \Id_n = n\sum_{i=1}^{m} A_{(i)}\otimes A_{(i)} - \Id_n = (n-1)\Id_n $, hence $\rho^2 \leq  n$. Algorithm~\ref{alg:derand:Bernstein} with $t=\OO(n\ln n /\eps^2)$ returns indices $x_1^*,x_2^*,\dots, x_t^*$ so that $\norm{\frac1{t} \sum_{j=1}^t f_j(x_j^*)} \leq \frac{\gamma \ln ( 2n)}{t \eps} + \eps \rho^2 /\gamma \leq 2\eps$. We next prove by induction that the same set of indices are also returned by Algorithm~\ref{alg:fast:isotrop}.
For ease of presentation, rescale every row of the input matrix $A$, i.e., set $ \widehat{A}_{(k)} = A_{(k)} \sqrt{ \theta / p_{k}}$ for every $k\in{[m]}$ (see Steps $2$ and $3$ of Algorithm~\ref{alg:fast:isotrop}). For sake of the analysis, let us define the following sequence of self-adjoint matrices of size $n$
\begin{eqnarray*}
T_{\{0\}} &:=& \zeromtx_n,\\
T_{\{i\}} &:=&  T_{\{i - 1\}} + \widehat{A}_{(x_i^*)} \otimes \widehat{A}_{(x_i^*)}  \text{ for } i\in{[t]}
\end{eqnarray*}
with eigenvalue decompositions $T_{\{i\}} = Q_{\{i\}} \Lambda_{\{i\}} Q^\top_{\{i\}} $, where $\Lambda_{\{i\}}$ are diagonal matrices containing the eigenvalues and the columns of $Q_{\{i\}}$ contain the corresponding eigenvectors. Set $Q_{\{0\}}=\Id$ and $\Lambda_{\{0\}}=\zeromtx$.
Notice that for every $k\in{[m]}$, by the eigenvalue decomposition of $T_{\{ i - 1\}}$, $T_{\{ i - 1\}} +  \widehat{A}_{(k)}\otimes \widehat{A}_{(k)} = Q_{\{i-1\}}\left(\Lambda_{\{i-1\}} + Q_{\{i-1\}}^\top \widehat{A}_{(k)} \otimes Q_{\{i-1\}}^\top \widehat{A}_{(k)}\right) Q_{\{i - 1\}}^\top.$ Observe that the above matrix (left hand side) and $\Lambda_{\{i-1\}} + Q_{\{i-1\}}^\top \widehat{A}_{(k)} \otimes Q_{\{i-1\}}^\top \widehat{A}_{(k)}$ have the same eigenvalues, since they are similar matrices. Let $\Lambda_{\{i-1\}} + Q_{\{i-1\}}^\top \widehat{A}_{(x_i^*)} \otimes Q_{\{i-1\}}^\top \widehat{A}_{(x_i^*)}  = U_{\{i\}} \Lambda_{\{i\}} U_{\{i\}}^\top$ be its eigenvalue decomposition\footnote{by its definition, $T_{\{i\}}$ has the same eigenvalues with $\Lambda_{\{i-1\}} + Q_{\{i-1\}}^\top \widehat{A}_{(x_i^*)} \otimes Q_{\{i-1\}}^\top \widehat{A}_{(x_i^*)}$.}. Then
\begin{eqnarray*}
T_{\{i - 1\}} + \widehat{A}_{(x_i^*)} \otimes \widehat{A}_{(x_i^*)} &  =  & Q_{\{i-1\}}\left(\Lambda_{\{i-1\}} + Q_{\{i-1\}}^\top \widehat{A}_{(x_i^*)} \otimes Q_{\{i-1\}}^\top \widehat{A}_{(x_i^*)}\right) Q_{\{i - 1\}}^\top\\ &  =  & Q_{\{i-1\}}U_{\{i\}} \Lambda_{\{i\}} U_{\{i\}}^\top  Q_{\{i - 1\}}^\top.
\end{eqnarray*}
It follows that $Q_{\{i\}} = Q_{\{i-1\}} U_{\{i\}}$ for every $i\geq 1$, so $Q_{\{i\}} = U_{\{1\}} U_{\{2\}} \dots  U_{\{i\}}$. The base case of the induction is immediate. Now assume that Algorithm~\ref{alg:fast:isotrop} has returned the same indices as Algorithm~\ref{alg:derand:Bernstein} up to the $(i-1)$-th iteration. It suffices to prove that at the $i$-th iteration Algorithm~\ref{alg:fast:isotrop} will return the index $x_i^*$.
We start with the expression in Step $4$ of Algorithm~\ref{alg:derand:Bernstein} and prove that it's equivalent (up to a fixed multiplicative constant factor) with the expression in Step $5$ of Algorithm~\ref{alg:fast:isotrop}. Indeed, for any $k\in{[m]}$, (let $C := \theta\sum_{j=1}^{i-1} f(x_j^*)$)
\begin{eqnarray*}
2\trace{\coshm{ C + \theta f(k) }}  =  \trace{\expm{ C  + \theta f(k) } + \expm{-C - \theta f(k) }}  \\
																   =  \trace{\expm{ T_{\{i-1\}} + \widehat{A}_{(k)}\otimes \widehat{A}_{(k)} }\e^{- \theta i} + \expm{- T_{\{i-1\}} - \widehat{A}_{(k)}\otimes \widehat{A}_{(k)} }\e^{ \theta i}}
\end{eqnarray*}
where we used the definition of $\coshm{\cdot}$, $f(i)$ and $T_{\{i-1\}}$ and the fact that the matrices commute. In light of Algorithm~\ref{alg:fast:isotrop} and the induction hypothesis, observe that the $m\times n$ matrix $Z$ at the start of the $i$-th iteration of Algorithm~\ref{alg:fast:isotrop} is equal to $\widehat{A} U_{\{1\}} U_{\{2\}} \ldots U_{\{i -1 \}}=  \widehat{A} Q_{\{i - 1\}}$. Now, multiply the latter expression that appears inside the trace with $Q_{\{i-1\}}^\top $ from the left and $Q_{\{i-1\}}$ from the right, it follows that ((let $C := \theta\sum_{j=1}^{i-1} f(x_j^*)$))
\begin{eqnarray*}
	2\trace{\coshm{ C + \theta f(k) }} & = & \trace{\expm{ \Lambda_{\{i-1\}} + Z_{(k)}\otimes Z_{(k)} }\e^{- \theta  i} + \expm{- \Lambda_{\{i-1\}} - Z_{(k)}\otimes Z_{(k)} }\e^{ \theta i}}
\end{eqnarray*}
using that $Q_{\{i-1\}}$ are the eigenvectors of $T_{\{i - 1\}}$ and the cyclic property of trace. This concludes part (\emph{a}).
Next we discuss how to deal with the technicality that arises from the approximate computation of the $\argmin$ expression in Step $5$ of Algorithm~\ref{alg:fast:isotrop}. First, let's assume that we have approximately (by invoking Lemma~\ref{lem:comp_eigs}) minimized the potential function in Step $5$ of Algorithm~\ref{alg:fast:isotrop}; denote this sequence of potential function values by  $\widetilde{\Phi}^{(1)},\ldots , \widetilde{\Phi}^{(t)}$. Next, we sufficiently bound the parameter $b$ of Lemma~\ref{lem:comp_eigs} so that the above approximation will not incur a significant multiplicative error.
Recall that at every iteration, by Ineq.~\eqref{ineq:barrier_incr} there exists an index over $[m]$ such that the current value of the potential function increases by at most a multiplicative factor $\exp\left( \eps^2 \rho^2 / \gamma^2\right)$. Lemma~\ref{lem:technical_cosh} tells us that at every iteration of Algorithm~\ref{alg:fast:isotrop} we increase the value of the potential function (by not selecting the optimal index over $[m]$) by at most an \emph{extra} multiplicative factor $\e^{2\delta}$, where $\delta$ is the additive error when computing the eigenvalues of the matrix in Step $5$ via Lemma~\ref{lem:comp_eigs}. Therefore, it follows that $\widetilde{\Phi}^{(t)} \leq \exp( 2\delta t) \Phi^{(t)}.$
Observe that, at the $i$-th iteration we apply Lemma\ref{lem:comp_eigs} on a matrix $\sum_{j=1}^{i} \widehat{A}_{(x_j)} \otimes \widehat{A}_{(x_j)}$ for some indices $x_j\in{[m]}$ and moreover $\norm{\sum_{j=1}^{i}{ \widehat{A}_{(x_j)} \otimes \widehat{A}_{(x_j)}} } = $ \\ $ \norm{ \theta \sum_{j=1}^{i}{ A_{(x_j)} \otimes A_{(x_j)} / p_{x_j} } } =   \norm{\theta \sum_{j=1}^{i} f(x_j ) - \theta i \Id}$. Triangle inequality tells us that $\norm{ \sum_{j=1}^{i}{ \widehat{A}_{(x_j)} \otimes \widehat{A}_{(x_j)}} }$ is at most $2\gamma\theta t$ for every $i \in{[t]}$. It follows that $\delta$ is at most $2^{-b+1} \theta t \gamma$ where $b$ is specified in Lemma~\ref{lem:comp_eigs}. The above discussion suggests that by setting $b= \OO(\log( \theta \gamma t))=\OO(\log (n\log n /\eps^3))$ we can guarantee that the potential function $\widetilde{\Phi}^{(t)} \leq  2n \exp\left( 3t \eps^2 \right)$. This concludes part~(\emph{b}).
%

%
Finally, we conclude the proof by analyzing the running time of Algorithm~\ref{alg:fast:isotrop}. Steps $2$ and $3$ can be done in $\OO(mn)$ time. Step $5$ requires $\widetilde{\OO}(mn\log^2 n )$ operations by invoking $m$ times Lemma~\ref{lem:comp_eigs}. Steps $6$ can be done in $\OO(n^2)$ time and Step $7$ requires $\widetilde{\OO}(mn\log^2 n )$ operations by invoking Lemma~\ref{lem:fast_mm:gerasoulis}. In total, since the number of iterations is $\OO(n\log n /\eps^2)$, the algorithm requires $\widetilde{\OO}( mn^2 \log^3 n /\eps^2)$ operations.
\end{proof}
%
\begin{proof}(of Theorem~\ref{thm:sparsification:here})
	Assume without loss of generality that $A$ has full rank. Define $u_i = A^{-1/2}v_i$ and notice that $\sum_{i=1}^{m} u_i \otimes u_i =\Id_n$. Run Algorithm~\ref{alg:fast:isotrop} with input $\{u_i\}_{i\in{[m]}}$ and $\eps$ which returns $\{\tau_i\}_{i\leq m}$, at most $t=\OO(n\log n /\eps^2)$ of which are non-zero such that
\begin{equation}\label{eqn:spectral_sparse:inner}
\norm{\sum_{i=1}^{m} \tau_i u_i \otimes u_i - \Id_n } \leq \eps.
\end{equation}
Define $\widehat{A} = A^{1/2}\left(\sum_{i=1}^{m}\tau_i u_i \otimes u_i\right) A^{1/2} = \sum_{i=1}^{m}\tau_i v_i\otimes v_i$. Eqn.~\eqref{eqn:spectral_sparse:inner} is equivalent to $(1-\eps) \Id_n \preceq \sum_{i=1}^{m} \tau_i u_i\otimes u_i \preceq (1+\eps) \Id_n$. Conjugating the latter expression by $A^{1/2}$, see~\cite[\S~$7.7$]{book:matrix_analysis:HornJohnson}, we get that $ (1-\eps) A \preceq \widehat{A} \preceq (1+\eps) A.$ Apply~\cite[Theorem~3.1]{phdthesis:Srivastava:2010} on $\widehat{A}$ which outputs a matrix $\widetilde{A}=\sum_{i=1}^{m}s_i v_i\otimes v_i$ with non-negative weights $\{s_i\}_{i\in{[m]}}$ at most $\lceil n /\eps^2 \rceil$ of which are non-zero, such that $ (1-\eps)^2 \widehat{A} \preceq \widetilde{A} \preceq (1+\eps)^2 \widehat{A}.$ Using the positive semi-definite partial ordering, we conclude that $(1-\eps)^3 A \preceq \widetilde{A} \preceq (1+\eps)^3 A$.
\end{proof}
%
%
\begin{proof}(of Theorem~\ref{thm:matrix_sparse:slow})
By homogeneity, assume that $\norm{A} = 1$. Following the proof of~\cite{matrix:sparsification:IPL2011}, we can assume that w.l.o.g. all non-zero entries of $A$ have magnitude at least $\eps/(2n)$ in absolute value, otherwise we can zero-out these entries and incur at most an error of $\eps/2$ (see~\cite[\S~$4.1$]{matrix:sparsification:IPL2011}).

Consider the bijection $\pi$ between the sets $[n^2]$ and $[n]\times [n]$ defined by $\pi (l)  \mapsto ( \lceil l / n\rceil , (l - 1) \mod n + 1) $ for every $l\in[n^2]$. Let $E_{ij}\in\reals^{n\times n}$ be the all zeros matrix having one only in the $(i,j)$ entry. Set $h(l) = \dil{ \frac{A_{\pi (l)}}{p_{l}} E_{\pi (l)} - A }$ where $p_l = A_{\pi (l)}^2/\frobnorm{A}^2$ for every $l\in{[n^2]}$. Observe that $h(\cdot ) \in \Sym^{2n\times 2n}$. Let $X$ be a random variable over $[n^2]$ with distribution $p_l$, $l\in{[n^2]}$. The same analysis as in Lemmas~$2$ and $3$ of~\cite{matrix:sparsification:IPL2011} together with properties of the dilation map imply that $\norm{h(l)} \leq 4n\sr{A} /\eps$ for every $l\in{[n^2]}$, $\EE{h(X)}=\zeromtx_{2n}$, and $\norm{\EE h(X)^2} \leq n\sr{A}$.

Run Algorithm~\ref{alg:derand:Bernstein} with $h(\cdot )$ as above. Algorithm~\ref{alg:derand:Bernstein} returns at most $t=28 n  \ln (\sqrt{2}n)\sr{A}/\eps^2$ indices $x_1^*,x_2^*,\ldots x_t^*$ over $[n^2]$ using $\OO(n^6 \sr{A} \log n /\eps^2 )$ operations such that
	\begin{equation}\label{ineq:esoteric}
		\norm{ \frac1{t}\sum_{l=1}^{t} h(x_l^*)} \leq \eps / 2.
	\end{equation}
Set $\widetilde{A} : = \frac1{t} \sum_{l=1}^{t} A_{\pi (x_l^*)} /p_{x_l^*} E_{\pi (x_l^*)}$. Observe that $\widetilde{A}$ has at most $t$ non-zero entries. Now, by the definition of $h(\cdot)$ and properties of the  dilation map, it follows that Ineq.~\eqref{ineq:esoteric} is equivalent to $\norm{ \dil{\widetilde{A} - A }} \ =\ \norm{ \widetilde{A} - A } \leq \eps /2.$
\end{proof}
%
%
%
\begin{proof}(of Lemma~\ref{lem:sparsif:decomp})
The key identity is $CC^\top : = \sum_{l,k\in{[n]},\ l<k} C^{(l,k)} \otimes C^{(l,k)}$. Let $l,k\in{[n]}$ with $l<k$, it follows that 
\begin{eqnarray*}
C^{(l,k)} \otimes C^{(l,k)}   & = & \left(\sqrt{|A_{lk}|} e_l + \sign{A_{lk}}  \sqrt{|A_{lk}|} e_k \right)\left(\sqrt{|A_{lk}|} e_l + \sign{A_{lk}}  \sqrt{|A_{lk}|} e_k\right)^\top\\
 & = & |A_{lk}| e_l \otimes e_l + A_{lk} e_k \otimes e_k + A_{lk} e_k \otimes e_l + |A_{lk}| e_k \otimes e_k.
\end{eqnarray*}
Therefore
\begin{equation}\label{eqn:sparse_decomp}
 CC^\top = \sum_{l,k\in{[n]}:\ l< k }\left[|A_{lk}| e_l \otimes e_l + A_{lk} e_k \otimes e_k + A_{lk} e_k \otimes e_l + |A_{lk}| e_k \otimes e_k\right].
\end{equation}
Let's first prove the equality for the off-diagonal entries of Eqn~\eqref{eqn:sparsify_lemma}. Let $l<k$ and $l,k\in{[n]}$. By construction, the only term of the sum that contributes to the $(i,j)$ and $(j,i)$ entry of the right hand side of Eqn.~\eqref{eqn:sparse_decomp} is the term $C^{(i,j)} \otimes C^{(i,j)} $. Moreover, this term equals $|A_{ij}| e_i \otimes e_i + A_{ij} e_i \otimes e_j + A_{ij} e_j \otimes e_i + |A_{ij}| e_j \otimes e_j$. Since $A_{ij} = A_{ji}$ this proves that the off-diagonal entries are equal.
For the diagonal entries of Eqn.~\eqref{eqn:sparsify_lemma}, it suffices to prove that $(CC^\top)_{ii} = R_i$. First observe that the last two terms of the sum in the right hand side of~\eqref{eqn:sparse_decomp} do not contribute to any diagonal entry. Second, the first two terms contribute only when $l=i$ or $k=i$. In the case where $l=i$, the contribution of the sum equals to $\sum_{i<k} |A_{ik}|$. On the other case ($k=i$), the contribution of the sum is equal to $\sum_{l<i} |A_{li}|$. However, $A$ is self-adjoint so $A_{li} =A_{il}$ for every $l<i$. It follows that the total contribution is $\sum_{i<k} |A_{ik}| + \sum_{l<i} |A_{il}| = \sum_{j\neq i} |A_{ij}| = R_i$.
\end{proof}
%
%
\begin{proof}( of Theorem~\ref{thm:matrix_sparsif:rand})
In one pass over the input matrix $A$ normalize the entries of $A$ by $\norm{A}$, so assume without loss of generality that $\norm{A}=1$. Let $C$ be the $n\times m$ matrix guaranteed by Lemma~\ref{lem:sparsif:decomp}, where $m= \binom{n}{2}$, each column of $C$ is indexed by the ordered pairs $(i,j)$, $i<j$ and $A = CC^\top +\diag{A} - R$. By definition of $C$ and the hypothesis, we have that $\norm{CC^\top} = \norm{A - \diag{A} +R} \leq \norm{A} +\infnorm{A}\leq 2\sqrt{\theta}$ and $\frobnorm{C}^2 =2 \sum_{i,j} |A_{ij}| \leq 2 n \infnorm{A}\leq 2n \sqrt{\theta} $.
Consider the bijection between the sets $[m]$ and $\{(i,j)\ |\ i<j,\ i,j\in{[n]}\}$ defined by $\pi (l)  \mapsto ( \lceil l / n\rceil , (l-1) \mod n + 1) $. For each $l\in{[m]}$, set $p_l=\norm{C^{\pi (l)}}^2 / \frobnorm{C}^2$ and define $f(l):= C^{\pi(l)} \otimes C^{\pi(l)}/p_l - CC^\top$. Let $X$ be a real-valued random variable over $[m]$ with distribution $p_l$. It is easy to verify that $\EE{f(X)} = \zeromtx_n$, $\norm{f(l)} \leq 2 \frobnorm{C}^2$ for every $l\in{[m]}$. A direct calculation gives that $\norm{\EE{ f(X)^2}} \leq 2\frobnorm{C}^2\norm{CC^\top}$. Matrix Bernstein inequality (see~\cite{chernoff:matrix_valued:Tropp}) with $f(\cdot)$ as above ($\gamma = 4n\sqrt{\theta}$ and $\rho^2 = 8 n \theta$) tells us that if we sample $t=38 n\theta \ln(\sqrt{2}n) /\eps^2 $ indices $x_1^*, x_2^*,\ldots , x_t^*$ over $[m]$ then with probability at least $1-1/n$, $\norm{ \frac1{t} \sum_{j=1}^{t} f(x_j^*)} \leq \eps$. Now, set $\widetilde{C}\in\reals^{n\times t}$ where the $j$-th column of $\widetilde{C}^{(j)}$ equals $\frac1{\sqrt{t}} C^{\pi(x_j^*)}$. It follows that $\norm{ \frac1{t} \sum_{j=1}^{t} f(x_j^*)} = \norm{ \frac1{t} \sum_{j=1}^{t} C^{\pi(x_j^*)} \otimes C^{\pi(x_j^*)} - CC^\top} = \norm{\widetilde{C}\widetilde{C}^\top - CC^\top}$. Define $\widetilde{A} = \widetilde{C}\widetilde{C}^\top +\diag{A} - R $. First notice that $\norm{\widetilde{A} - A} = \norm{\widetilde{C} \widetilde{C}^\top -CC^\top} \leq \eps$. It suffices to bound the number of non-zeros of $\widetilde{A}$. To do so, view the matrix-product $\widetilde{C}\widetilde{C}^\top$ as a sum of rank-one outer-products over all columns of $\widetilde{C}$. By the special structure of the entries of $\widetilde{C}$, every outer-product term of the sum contributes to at most four non-zero entries, two of which are off-diagonal. Since $\widetilde{C}$ has at most $t$ columns, $\widetilde{A}$ has at most $n + 2t$ non-zero entries; $n$ for the diagonal entries and $2t$ for the off-diagonal.
\end{proof}
\begin{proof}(of Theorem~\ref{thm:matrix_sparsif:det})
Let $C$ be the $n\times m$ matrix such that $A = CC^\top +\diag{A} - R$ and $m\leq \nnz{A}$ guaranteed by Lemma~\ref{lem:sparsif:decomp}. Apply Theorem~\ref{thm:sparsification:here} on the matrix $CC^\top$  and $\eps$ which outputs, in deterministic $\widetilde{\OO}(\nnz{A} n^2 \theta \log^3 n  /\eps^2 + n^4 \theta^2 \log n /\eps^4)$ time, an $n\times \lceil n/\eps^2\rceil$ matrix $\widetilde{C}$ such that $(1-\eps)^3 CC^\top \preceq \widetilde{C}\widetilde{C}^\top \preceq (1+\eps)^3 CC^\top.$ By Weyl's inequality~\cite[Theorem~$4.3.1$]{book:matrix_analysis:HornJohnson} and the fact that $\eps<1/2$, it follows that $\norm{CC^\top - \widetilde{C}\widetilde{C}^\top} \leq 5 \eps \norm{CC^\top}$. Define $\widetilde{A}:= \widetilde{C}\widetilde{C}^\top + \diag{A} - R$. First we argue that the number of non-zero entries of $\widetilde{A}$ is at most $n+ \lceil 2n/\eps^2 \rceil $. Recall that every column of $\widetilde{C}$ is a rescaled column of $C$. Now, think the matrix-product $\widetilde{C}\widetilde{C}^\top$ as a sum of rank-one outer-products over all columns of $\widetilde{C}$. By the special structure of the entries of $\widetilde{C}$, every outer-product term of the sum contributes to at most four non-zero entries, two of which are off-diagonal. Since $\widetilde{C}$ has at most $\lceil n/\eps^2 \rceil$ columns, $\widetilde{A}$ has at most $n + \lceil 2 n/\eps^2\rceil$ non-zero entries; $n$ for the diagonal entries and $\lceil 2 n/\eps^2\rceil$ for the off-diagonal. Moreover, $\widetilde{A}$ is close to $A$ in the operator norm sense. Indeed,
\begin{eqnarray*}
	\norm{A - \widetilde{A} } &   =  &  \norm{CC^\top - \widetilde{C}\widetilde{C}^\top}  \leq\ 5\eps\norm{CC^\top } \   =  \  5\eps\norm{A - \diag{A} +R }\\
	  						  & \leq &  5\eps(\norm{A} + \infnorm{A}) \ \leq\ 10 \eps \sqrt{\theta}\norm{A}
\end{eqnarray*}
where we used the definition of $\widetilde{A}$, Eqn.~\eqref{eqn:sparsify_lemma}, triangle inequality, the assumption that $A$ is $\theta$-SDD and the fact that $\theta \geq 1$. Repeating the proof with $\eps' =\frac{\eps}{10\sqrt{\theta}}$ and elementary manipulations conclude the proof.
\end{proof}

\end{document}